\documentclass[aip, showpacs
	,jmp
	,reprint
	,onecolumn
	]{revtex4-1}

\usepackage{amsmath,amsthm,amssymb}

\usepackage{graphicx}

\def\R{\mathbb R}
\def\N{\mathbb N}
\def\C{\mathbb C}

\renewcommand{\d}{{\mathrm d}}
\renewcommand{\i}{{\mathrm i}}
\def\e{\mathrm e}

\def\P{\mathcal P}

\newtheorem{lemma}{Lemma}[section]
\newtheorem{proposition}[lemma]{Proposition}
\newtheorem{theorem}[lemma]{Theorem}
\newtheorem{corollary}[lemma]{Corollary}
\newtheorem{observation}[lemma]{Observation}

\theoremstyle{definition}

\theoremstyle{remark}
\newtheorem{remark}[lemma]{Remark}

\begin{document}

\title{Quantum graph as a quantum spectral filter}
\author{Ond\v{r}ej Turek}
\email[]{ondrej.turek@kochi-tech.ac.jp}
\homepage[]{\tt \\ http://researchmap.jp/turek/}
%
\author{Taksu Cheon}
\email[]{taksu.cheon@kochi-tech.ac.jp}
\homepage[]{\tt \\ http://researchmap.jp/T_Zen/}
\affiliation{Laboratory of Physics, Kochi University of Technology,
 Tosa Yamada, Kochi 782-8502, Japan}

\date{\today}

\begin{abstract}
We study the transmission of a quantum particle along a straight input--output line to which a graph $\Gamma$ is attached at a point. In the point of contact we impose a singularity represented by a certain properly chosen scale-invariant coupling with a coupling parameter $\alpha$. We show that the probability of transmission along the line as a function of the particle energy tends to the indicator function of the energy spectrum of $\Gamma$ as $\alpha\to\infty$. This effect can be used for a spectral analysis of the given graph $\Gamma$. Its applications include a control of a transmission along the line and spectral filtering. The result is illustrated with an example where $\Gamma$ is a loop exposed to a magnetic field. Two more quantum devices are designed using other special scale-invariant vertex couplings. They can serve as a band-stop filter and as a spectral separator, respectively.
\end{abstract}


\maketitle

\section{Introduction}

Quantum graphs serve as mathematical models of mesoscopic networks built from thin nano-sized wires. Such wires can be made of semiconductors, carbon and other materials. With respect to the current rapid development of nanotechnologies, quantum graphs have a considerable application potential. Theoretical literature on the subject is now very extensive~\cite{EKST08}. In this paper we focus on the use of quantum graphs for a design of quantum 
devices that allow to control the transmission of an electron along a line according to its energy.

Various types of 
filtering capabilities of quantum graphs are known already for some time. One of the simplest such examples is the $\delta$-interaction on a line, which works as a high-pass filter. 
There exist also devices built upon graphs with more edges. For instance, it has been shown that a star graph with three arms and a properly chosen point interaction in the vertex can work as a high-pass/low-pass junction~\cite{CET09}. Very recently, a star graph with three edges coupled together by a scale-invariant point interaction has been used to design a controllable band-pass spectral filter~\cite{TC11}. Its controllability is achieved by an external potential on one of the edges. This construction has been generalized to quantum filters with multiple outputs and with multiple controllers~\cite{TC12}. There exist also different designs, e.g., a special trident filter~\cite{SCdL03}.

In this paper we consider transmission characteristics of a graph built from four ``components'': an ``input'' half line, an ``output'' half line, a graph $\Gamma$, and a certain special scale-invariant point interaction. Both the half lines are attached to the graph $\Gamma$ in one of its vertices, and these three objects are coupled together by the scale-invariant interaction. After the preliminaries in Section~\ref{Section: Preliminaries}, we 
study, 
in Section~\ref{Section: Main}, the transmission along the input--output line of the graph in question. 
We find that the transmission characteristics show a strong resonance behavior 
at energies belonging to the spectrum of $\Gamma$.
Therefore, the graph can be regarded as a band-pass spectral filter with narrow peak passbands. However, it is a slight abuse of terminology to call the graph ``filter'', because the passbands are not of an interval-type, but of a peak-type around the resonance energies.
Naturally, this approach enables the design of spectral resonance filters with various characteristics, depending on the chosen graph $\Gamma$.

Section~\ref{Section: Loop} illustrates the result with an example. We choose $\Gamma$ as a loop placed in a magnetic field $B$. The strength of $B$ determines the spectrum of $\Gamma$ via a simple formula, and, in consequence, it directly controls the passbands of the filter. Therefore, the device can be used as a spectral filter controllable by an external magnetic field.

The 
resonance behavior
of the proposed filter 
essentially relies on the scale-invariant vertex coupling. Since the physical interpretation of this coupling is not straighforward, we devote Section~\ref{Section: Approximation} to an explanation how to obtain it approximately by a use of several $\delta$-interactions, which are better understood.

It turns out that there exist other special scale-invariant couplings that enable the construction of quantum devices with other interesting characteristics. We discuss two such examples in Sections~\ref{Section: Inverse} and \ref{Section: Switch}.
In Section~\ref{Section: Inverse} we find a scale-invariant coupling that allows to build a band-stop resonance filter with stopbands located at the energies belonging to the spectrum of $\Gamma$. In other words, its transmission characteristics are complementary to the characteristics of the filter designed in Section~\ref{Section: Main}. In Section~\ref{Section: Switch} we consider a device with two outputs that works as a spectral separator. Broadly speaking, particles with energies outside the spectrum of $\Gamma$ are transmitted to output 1, while particles with energies from the spectrum of $\Gamma$ are transmitted to output 2. If the spectrum of $\Gamma$ is governed by an external field, the device can be used as a controllable switch or spectral junction.

\section{Preliminaries}\label{Section: Preliminaries}

Let $\Gamma$ be a graph, $V_\Gamma$ denote the set of its vertices and $E_\Gamma$ the set of its edges. We assume that $\Gamma$ is a metric graph, i.e., every edge $e\in E_\Gamma$ has its length $\ell_e>0$. If $n$ is the cardinality of $E_\Gamma$, then the wave function of a particle on $\Gamma$ has $n$ components: $\Psi=(\psi_1,\ldots,\psi_n)^T$, where the superscript $T$ stands for the transposition. Let there be scalar potentials $U_1,\ldots,U_n$ and vector potentials $A_1,\ldots,A_n$ on the graph edges. The Hamiltonian of a particle on $\Gamma$, denoted by $H_{\Gamma}$, acts as
$$
H_{\Gamma}
\begin{pmatrix}
\psi_1\\
\vdots\\
\psi_n
\end{pmatrix}
=\frac{1}{2m}
\begin{pmatrix}
\left(-\i\hbar\frac{\d}{\d x}-qA_1\right)^2\psi_1+U_1\cdot\psi_1\\
\vdots\\
\left(-\i\hbar\frac{\d}{\d x}-qA_n\right)^2\psi_n+U_n\cdot\psi_n
\end{pmatrix}
$$
for $\Psi\in L^2(\Gamma)$, where $m$ is the mass of the particle and $q$ is its charge.

In order to make the operator $H_{\Gamma}$ self-adjoint, it is necessary to impose proper boundary conditions in the graph vertices.
Let $v\in V_\Gamma$ be a vertex of degree $\deg(v)$ and $\psi_1,\ldots,\psi_{\deg(v)}$ be the wave function components at the edges incident to $v$. If $\psi_1(0),\ldots,\psi_{\deg(v)}(0)$ are the limits of those components in the vertex $v$ and $\psi_1'(0),\ldots,\psi_{\deg(v)}'(0)$ are the limits of their derivatives in $v$, taken in the outgoing sense, we denote
\begin{equation}
\Psi_v=
\begin{pmatrix}
\psi_1(0) \\
\vdots \\
\psi_{\deg(v)}(0)
\end{pmatrix}
\quad\text{and}\quad
\Psi_v'=
\begin{pmatrix}
\psi'_1(0) \\
\vdots \\
\psi'_{\deg(v)}(0)
\end{pmatrix} \,.
\end{equation}
The boundary conditions at every vertex $v$ couple $\Psi_v$ and $\Psi'_v$ in the way
\begin{equation}\label{b.c.}
A\Psi_v+B\Psi'_v=0\,,
\end{equation}
where $A$ and $B$ are complex $\deg(v)\times \deg(v)$ matrices such that~\cite{KS99}
\begin{equation}\label{KS}
\mathrm{rank}(A|B)=\deg(v) \quad\text{and}\quad AB^*=(AB^*)^*\,.
\end{equation}
The symbol $(A|B)$ denotes the $\deg(v)\times2\deg(v)$ matrix with $A,B$ forming
the first and the second $\deg(v)$ columns, respectively.

The requirements~\eqref{KS} are essentially equivalent to certain explicit constraints imposed on the matrix pair $(A,B)$. In this paper we will take advantage of the so-called $ST$-form~\cite{CET10}. It consists in expressing the boundary conditions~\eqref{b.c.} in the block form
\begin{equation}\label{ST}
\left(\begin{array}{cc}
I^{(r)} & T \\
0 & 0
\end{array}\right)\Psi'_v=
\left(\begin{array}{cc}
S & 0 \\
-T^* & I^{(\deg(v)-r)}
\end{array}\right)\Psi_v\,,
\end{equation}
where $r\in\{0,1,\ldots,\deg(v)\}$, $I^{(\deg(v))}$ is the identity matrix of size $\deg(v)$, $T$ is a general complex $r\times(\deg(v)-r)$ matrix and $S$ is a Hermitian $r\times r$ matrix.

If $S$ in~\eqref{ST} is a zero matrix, then the boundary conditions do not mix the values of functions and of their derivatives. Boundary conditions of that type are usually called scale-invariant vertex conditions~\cite{FKW07}. They define an interesting family of vertex couplings~\cite{FT00, NS00, SS02} with useful scattering properties~\cite{TC11,TC12,CT10}.

One of the most natural singular interactions is the $\delta$-coupling (also called ``$\delta$~potential''), which is characterized by boundary conditions
\begin{eqnarray}\label{delta}
&
\psi_j(0)=\psi_\ell(0)=:\psi(0) \quad \forall j,\ell=1\ldots,\deg(v)\,, 
\nonumber \\
&
\sum^{\deg(v)}_{j=1}\psi_j'(0)=\alpha\psi(0)\,,
\end{eqnarray}
where $\alpha\in\R\backslash\{0\}$ is the parameter of the coupling. The $\delta$-coupling does not belong to the scale-invariant family, but is prominent due to its simple interpretation:
It can be understood as a limit case of properly scaled smooth potentials~\cite{Ex96b}.

If we set $\alpha=0$ in~\eqref{delta}, we obtain boundary conditions of the \emph{free coupling},
\begin{eqnarray}\label{free}
&
\psi_j(0)=\psi_\ell(0) \quad \forall j,\ell=1\ldots,\deg(v)\,, 
\nonumber \\
&
\sum^{\deg(v)}_{j=1}\psi_j'(0)=0\,,
\end{eqnarray}
which is the most trivial type of point interaction in a quantum graph.

\section{Transmission along a line with an attached graph}\label{Section: Main}

Let $\Gamma=(V_\Gamma,E_\Gamma)$ be a finite connected metric graph. We denote the cardinality of $E_\Gamma$ by $n$ for the sake of brevity. From now on let $\Phi=(\phi_1,\ldots,\phi_n)^T$ denote the wave function on $\Gamma$. We assume that the graph edges are finite, the potentials on the graph edges are bounded and that the self-adjoint boundary conditions in the vertices of $\Gamma$ are chosen in the following way:
\begin{itemize}
\item There is a vertex $v_0\in V_\Gamma$ with free boundary conditions~\eqref{free}.
\item In all the remaining vertices $v\in V_\Gamma\backslash\{v_0\}$ we admit any self-adjoint boundary conditions~\eqref{b.c.}\&\eqref{KS} except for those virtually decoupling adjacent edges.
\end{itemize}

\begin{figure}[h]
\begin{center}
\includegraphics[width=4.2cm]{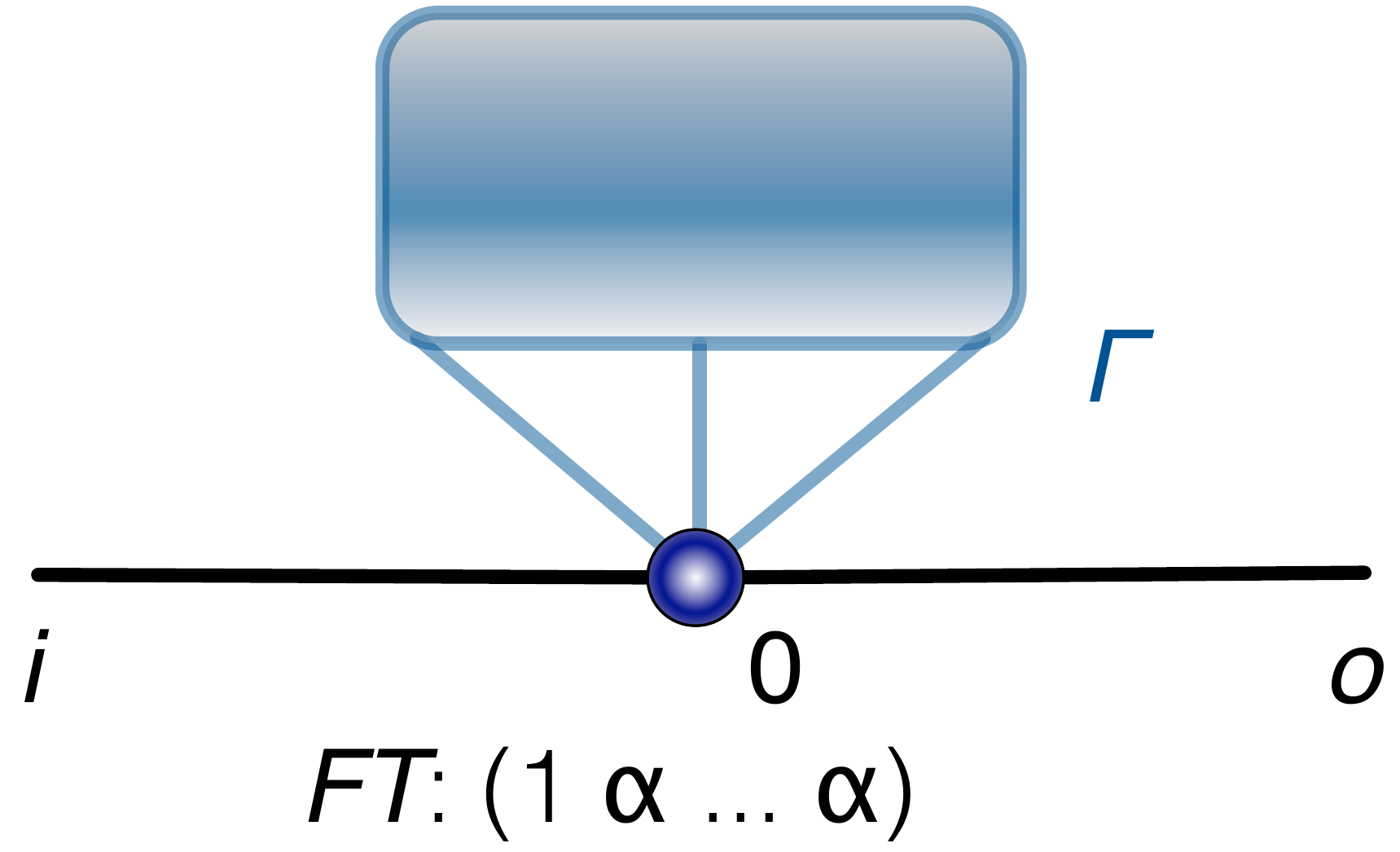}
\caption{A line with an attached graph $\Gamma$.}
\label{Fig: graph}
\end{center}
\end{figure}
Note that the presence of a vertex with the free coupling in the graph $\Gamma$ can be assumed without loss of generality. For example, any point inside a graph edge has this property, and, therefore, can be regarded as the vertex $v_0$.

Now let us consider a graph $\Gamma_{\mathbf{io}}$ which is constructed from $\Gamma$ by attaching two half lines to its vertex $v_0$ (Figure~\ref{Fig: graph}). We denote the half lines by $\mathbf{i}$ (``input'') and $\mathbf{o}$ (``output''). Furthermore, we use a new symbol for the vertex of $\Gamma_{\mathbf{io}}$ created from the vertex $v_0$ of $\Gamma$; it is convenient to denote it by $0$. Therefore, $\Gamma_{\mathbf{io}}=(\{0\}\cup V_\Gamma\backslash\{v_0\},E_\Gamma\cup\{\mathbf{i},\mathbf{o}\})$. Graph $\Gamma_{\mathbf{io}}$ can be also regarded as a straight ``input--output'' line, parametrized by $x\in(-\infty,+\infty)$, to which $\Gamma$ is attached in the point $0$.
The wave function component on the input half line will be denoted by $\psi_-$ and parametrized by $x\in(-\infty,0)$. The wave function component on the output half line will be denoted by $\psi_+$ and parametrized by $x\in(0,+\infty)$.

The half lines $\mathbf{i}$ and $\mathbf{o}$ carry no potentials. Therefore,
the Hamiltonian $H$ on $\Gamma_{\mathbf{io}}$ acts as

\begin{equation}\label{H}
H
\begin{pmatrix}
\psi_- \\
\psi_+ \\
\phi_1\\
\vdots\\
\phi_n
\end{pmatrix}
=\frac{1}{2m}
\begin{pmatrix}
-\hbar^2\psi_-'' \\
-\hbar^2\psi_+'' \\
\left(-\i\hbar\frac{\d}{\d x}-qA_1\right)^2\phi_1+U_1\cdot\phi_1\\
\vdots\\
\left(-\i\hbar\frac{\d}{\d x}-qA_n\right)^2\phi_n+U_n\cdot\phi_n
\end{pmatrix}
.
\end{equation}
We assume that the boundary conditions in each vertex $v\in V_\Gamma\backslash\{v_0\}$ of $\Gamma_{\mathbf{io}}$ are the same as in the corresponding vertex $v$ of $\Gamma$.
In the vertex $0$ of $\Gamma_{\mathbf{io}}$, we impose a coupling given by the scale-invariant boundary conditions
\begin{widetext}
\begin{equation}\label{bc}
\begin{pmatrix}
1 & 1 & \alpha & \cdots & \alpha \\
0 & 0 & 0 & \cdots & 0 \\
0 & 0 & 0 & \cdots & 0 \\
\vdots & \vdots & \vdots &  & \vdots \\
0 & 0 & 0 & \cdots & 0
\end{pmatrix}
\begin{pmatrix}
-\psi_-'(0) \\
\psi_+'(0) \\
\phi_1'(0) \\
\vdots \\
\phi_{n}'(0)
\end{pmatrix}
=
\begin{pmatrix}
0 & 0 & 0 & \cdots & 0 \\
-1 & 1 & 0 & \cdots & 0 \\
-\alpha & 0 & 1 &  & 0 \\
\vdots & \vdots &  & \ddots & \\
-\alpha & 0 & 0 &  & 1
\end{pmatrix}
\begin{pmatrix}
\psi_-(0) \\
\psi_+(0) \\
\phi_1(0) \\
\vdots \\
\phi_{n}(0)
\end{pmatrix}\,,
\end{equation}
\end{widetext}
where
\begin{itemize}
\item $\alpha>0$ is a parameter of the coupling,
\item $n=\deg(v_0)$,
\item functions $\phi_1(x),\ldots,\phi_{n}(x)$ are the wave function components on those edges $e\in E_\Gamma$ which are incident to the vertex $0$.
\end{itemize}
The left derivative of $\psi_-$ is taken with the minus sign, because the limits of derivatives of the wave function components in the graph vertices are conventionally considered in the outgoing sense.

It is convenient to rewrite the boundary conditions~\eqref{bc} as a set of equations:
\begin{subequations}\label{bc_sub}
\begin{gather}
-\psi_-'(0)+\psi_+'(0)+\alpha\sum_{j=1}^{n}\phi_j'(0)=0\,; \label{bc1} \\
\psi_-(0)=\psi_+(0); \label{bc2} \\
\alpha\psi_-(0)=\phi_j(0) \qquad\text{for all $j=1,\ldots,n$}\,. \label{bc3}
\end{gather}
\end{subequations}

\subsection*{Transmission along the input--output line}
Let us consider a particle of energy $E>0$ moving along the input half line towards the vertex $0$. When the particle reaches the vertex,
it is scattered into all incident edges. Therefore, the final-state wave function components on $\mathbf{i}$ and $\mathbf{o}$ take the form
\begin{subequations}\label{psi_io}
\begin{align}
\psi_-(x)&=\e^{ikx}+\mathcal{R}(k)\e^{-ikx}\,, \label{psi_i} \\
\psi_+(x)&=\mathcal{T}(k)\e^{ikx}\,, \label{psi_o}
\end{align}
\end{subequations}
where
\begin{equation}\label{k}
k=\frac{\sqrt{2mE}}{\hbar}
\end{equation}
is the wavenumber at the input--output line. The coefficient $\mathcal{R}(k)$ represents the reflection amplitude, and $\mathcal{T}(k)$ is the amplitude of transmission of the particle from the input half line $\mathbf{i}$ to the output half line $\mathbf{o}$.
The value $|\mathcal{T}(k)|^2$ represents the probability of transmission from the input half line to the output half line for the given wavenumber $k$. From now on we denote this probability by $\P(k)$.

Consider the following problem~\cite{SA00,Ku05}:
\begin{subequations}\label{Problem}
\begin{align}
&\frac{1}{2m}\left[\left(-\i\hbar\frac{\d}{\d x}-qA_j\right)^2\phi_j+U_j\cdot\phi_j\right]=\lambda\phi_j 
\nonumber \\
&\qquad
\forall j=1,\ldots,n\,, \label{Problem 1} \\
&\text{$\phi_1,\ldots,\phi_n$ satisfy the boundary conditions} \nonumber\\
 &\qquad\qquad \text{in every
$v\in V_\Gamma\backslash\{v_0\}$}\,, \label{Problem 2}\\
&\phi_j(0)=1 \qquad \text{for all $j=1,\ldots,n$}\,. \label{Problem 3}
\end{align}
\end{subequations}
Let us define
$$
\sigma_0=\{\lambda\in(0,+\infty)\;|\;\text{problem~\eqref{Problem} has no solution}\}\,.
$$

\begin{observation}\label{Obs. sigma0}
\begin{itemize}
\item If $\lambda\in\rho(H_{\Gamma})$, the problem~\eqref{Problem} has a unique solution. Hence $\sigma_0\subset\sigma(H_{\Gamma})$.
\item If $\lambda\in\sigma(H_{\Gamma})\backslash\sigma_0$, then
every solution of~\eqref{Problem} is an eigenfuction of $H_\Gamma$ corresponding to the eigenvalue $\lambda$. Consequently, $\sum_{j=1}^{n}\phi_j'(0)=0$.
\item If $\lambda\in\sigma_0$, then every eigenfunction of $H_{\Gamma}$ corresponding to the eigenvalue $\lambda$ satisfies $\phi_1(0)=\cdots=\phi_{n}(0)=0$.
\end{itemize}
\end{observation}

The set $\sigma_0$ can be equivalently characterized by the condition
\begin{equation}\label{sigma0}
\lambda\in\sigma_0 \  \Leftrightarrow \  \left[\; \lambda\in\sigma(H_{\Gamma}) \  \wedge \  \left(\, H_{\Gamma}\Phi=\lambda\Phi \,\Rightarrow\, \phi_1(0)=0 \, \right) \;\right]\,.
\end{equation}
For every $\lambda\in(0,+\infty)\backslash\sigma_0$, we define the Dirichlet-to-Neumann function~\cite{Ku05,Ca11,SU90} as
\begin{equation}\label{Lambda}
\Lambda(\lambda)=\sum_{j=1}^{n}\phi_j'(0)\,,
\end{equation}
where $(\phi_1,\ldots,\phi_n)$ is a solution of the problem~\eqref{Problem}.
The function $\Lambda$ is well-defined with regard to Observation~\ref{Obs. sigma0}. Note that $\Lambda(\lambda)=0$ together with~\eqref{Problem 3} means that $\Phi:=(\phi_1,\ldots,\phi_n)^T$ obeys the free boundary conditions in $v_0$. Hence we obtain:
\begin{observation}\label{Obs. Lambda=0}
For every $\lambda\in(0,+\infty)\backslash\sigma_0$,
$$
\Lambda(\lambda)=0 \quad\Leftrightarrow\quad \lambda\in\sigma(H_{\Gamma})\,.
$$
\end{observation}

\begin{proposition}\label{Prop. T}
Let $\mathcal{T}(k)$ be the amplitude of the transmission to the output line for an incoming particle of energy $E=\frac{\hbar^2k^2}{2m}$.
It holds:
\begin{itemize}
\item[(i)] If $E\in(0,+\infty)\backslash\sigma_0$, then
\begin{equation}\label{T}
\mathcal{T}(k)=\frac{1}{1+\alpha^2\frac{\Lambda(E)}{2\i k}}\,.
\end{equation}
\item[(ii)] If $E\in\sigma_0$, then $\mathcal{T}(k)=0$.
\end{itemize}
\end{proposition}

\begin{proof}

(i)\quad If $E\in(0,+\infty)\backslash\sigma_0$, the problem~\eqref{Problem} has a solution $\Phi=(\tilde{\phi}_1,\ldots,\tilde{\phi}_{n})^T$. We set
\begin{equation}\label{_Psi}
\Psi:=(\e^{\i kx}+\mathcal{R}(k)\e^{-\i kx},\mathcal{T}(k)\e^{\i kx},c\tilde{\phi}_1,\ldots,c\tilde{\phi}_{n})^T\,,
\end{equation}
where $\mathcal{R}(k)$ and $\mathcal{T}(k)$ are the sought scattering amplitudes and $c\in\C$ is a constant to be specified later. The function $\Psi$ obviously satisfies the system of differential equations $H\Psi=E\Psi$ due to~\eqref{Problem 1}. Moreover, $\Psi$ obeys the boundary conditions in each $v\in V_\Gamma\backslash\{v_0\}$ due to the assumption~\eqref{Problem 2}. To sum up, $\Psi$ is the final-state wave function on $\Gamma_{\mathbf{io}}$ if and only if $\Psi$ obeys the boundary conditions~\eqref{bc_sub} in the vertex $0$. We rewrite the boundary conditions~\eqref{bc_sub} using the properties of $\Phi$, cf.~\eqref{Problem}:
\begin{gather*}
\i k(-1+\mathcal{R}(k)+\mathcal{T}(k))+\alpha c\Lambda(E)=0\,; \\
1+\mathcal{R}(k)=\mathcal{T}(k); \\
\alpha(1+\mathcal{R}(k))=c \qquad\text{for all $j=1,\ldots,n$}\,.
\end{gather*}
This system yields
$c=\frac{2\i k\alpha}{2\i k+\alpha^2\Lambda(E)}$ and $\mathcal{T}(k)=\frac{1}{1+\alpha^2\frac{\Lambda(E)}{2\i k}}$.

(ii)\quad If $E\in\sigma_0$, let $\Psi=(\e^{\i kx}+\mathcal{R}(k)\e^{-\i kx},\mathcal{T}(k)\e^{\i kx},\phi_1,\ldots,\phi_{n})^T$ be the sought final-state wave function on $\Gamma_{\mathbf{io}}$. Then the $n$-tuple $\phi_1,\ldots,\phi_{n}$ obviously satisfies~\eqref{Problem 1} and \eqref{Problem 2}. Moreover, $\Psi$ obeys the boundary conditions~\eqref{bc3}, hence $\phi_1(0)=\cdots=\phi_{n}(0)$. The assumption $E\in\sigma_0$ means that the problem~\eqref{Problem} does not have a solution, therefore necessarily $\phi_1(0)=\cdots=\phi_{n}(0)=0$. Then $0=1+\mathcal{R}(k)=\mathcal{T}(k)$ according to~\eqref{bc3} and \eqref{bc2}.

\end{proof}

Proposition~\ref{Prop. T} together with Observation~\ref{Obs. Lambda=0} allows to calculate the limit of the transmission probability $\P(k)$ for the coupling parameter $\alpha\to\infty$.

\begin{corollary}\label{T coro}
For every $k>0$,
$$
\lim_{\alpha\to\infty}\P(k)=
\left\{\begin{array}{cl}
1 & \text{if $\frac{\hbar^2k^2}{2m}\in\sigma(H_{\Gamma})\backslash\sigma_0$}\,, \\
0 & \text{otherwise}\,.
\end{array}\right.
$$
\end{corollary}

Therefore, if $\alpha\gg1$, the transmission probability is a function with sharp peaks attaining $1$ located just at the points $k$ corresponding to energies $E\in\sigma(H_{\Gamma})\backslash\sigma_0$.
It means that this quantum device works as a band-pass spectral filter.

\begin{remark}\label{Rem. sigma0}
The set $\sigma_0$ can be made empty by a convenient choice of the vertex $v_0$ in the graph $\Gamma$; usually it suffices to take $v_0$ as a random point inside an edge of $\Gamma$. Then $\P(k)$ converges to the characteristic (indicator) function of the set $\{k>0\,|\,\frac{\hbar^2k^2}{2m}\in\sigma(H_\Gamma)\}$ in the limit $\alpha\to\infty$.
\end{remark}

\begin{remark}
Notice that the validity of the result obtained in Proposition~\ref{Prop. T} does not rely on the fact that $\Gamma$ is a graph. The proposition can be formulated in a similar way for $\Gamma$ being a structure comprising one-, two- and three-dimensional objects which is attached to the input--output line via some of its one-dimensional lines (``antennas'')~\cite{ES97}.
Also graphs $\Gamma$ with infinite edges can be considered~\cite{TC11,TC12}.
\end{remark}

A natural application of the phenomenon arises in spectral filtering. Let $\Gamma_{\mathbf{io}}$ be constructed for $\alpha\gg1$. Let particles of various energies be sent along the input line to the vertex $0$. Then particles with $E\approx\lambda\in\sigma(H_{\Gamma})$ pass through the vertex $0$ to the output line, whereas particles of other energies are reflected or deflected to the graph $\Gamma$. If moreover the spectrum of $\Gamma$ can be adjusted by external fields, we obtain a controllable spectral filter. An example will be studied in the next section.

\section{Quantum spectral filter controlled by magnetic field}\label{Section: Loop}

Now we consider a concrete example of the filter developed in Section~\ref{Section: Main}. In this example, $\Gamma$ is an Aharonov-Bohm ring \cite{AB59}, a loop of length $\ell$, as depicted in Figure~\ref{Fig: loop}.
We denote the wave function component on the loop by $\phi$ and parametrize it by $x\in(0,\ell)$. The components on the input half line and on the output half line are denoted by $\psi_-$ and $\psi_+$, respectively, in accordance with Section~\ref{Section: Main}.
\begin{figure}[h]
\begin{center}
\includegraphics[width=4.2cm]{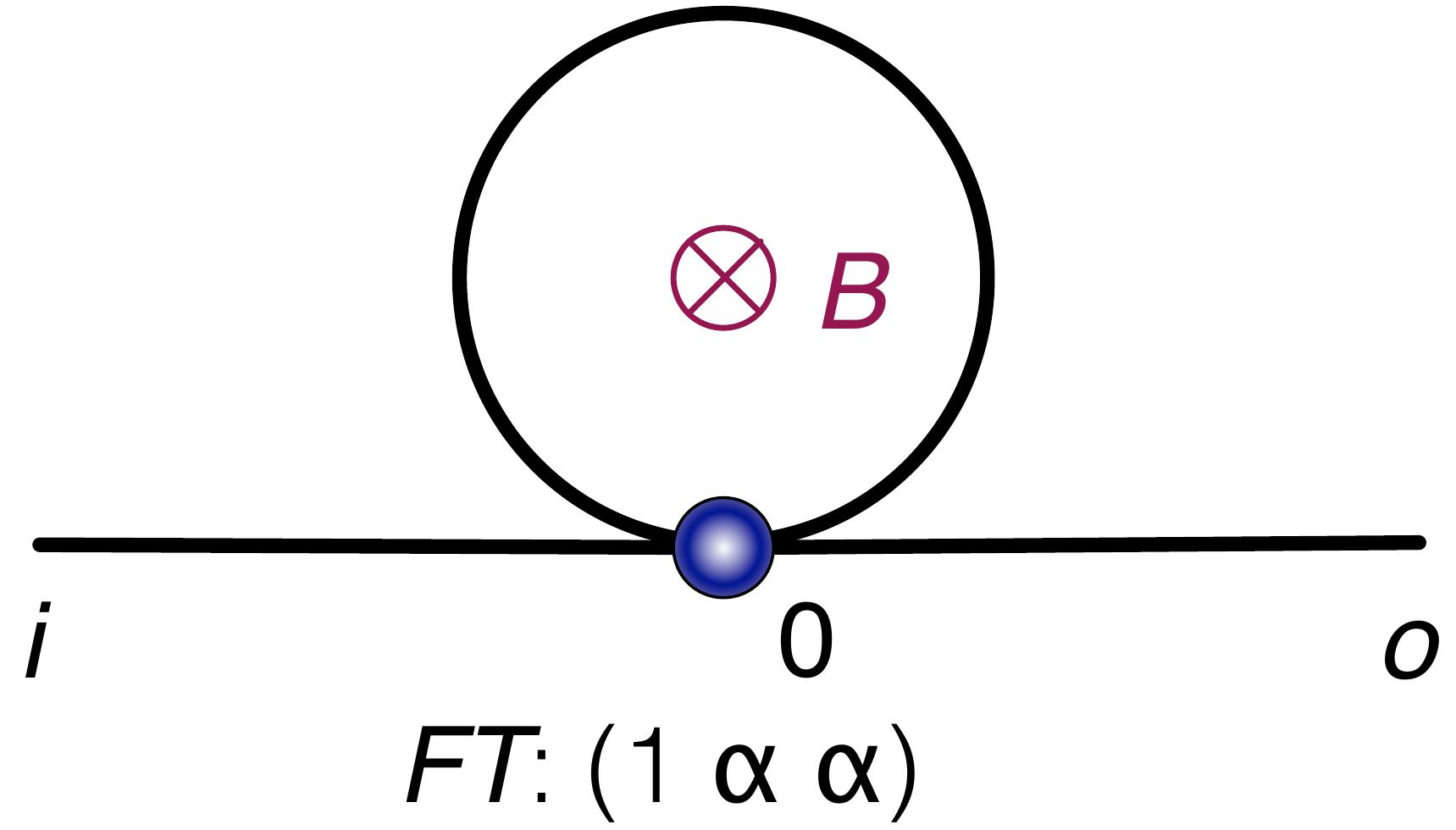}
\caption{A line with an attached loop.}
\label{Fig: loop}
\end{center}
\end{figure}

The vertex $0$, in which the input-output line and the loop are connected, has degree $4$. The scale-invariant coupling~\eqref{bc} in $0$ is thus given by the boundary conditions
\begin{equation}\label{bc loop}
\begin{pmatrix}
  1 & 1 & \alpha & \alpha \\
  0 & 0 & 0 & 0 \\
  0 & 0 & 0 & 0 \\
  0 & 0 & 0 & 0
\end{pmatrix}
\begin{pmatrix}
  -\psi_-'(0) \\
  \psi_+'(0) \\
  \phi'(0) \\
  -\phi'(\ell)
\end{pmatrix}
=
\begin{pmatrix}
  0 & 0 & 0 & 0 \\
 -1 & 1 & 0 & 0 \\
 -\alpha & 0 & 1 & 0 \\
 -\alpha & 0 & 0 & 1
\end{pmatrix}
\begin{pmatrix}
  \psi_-(0) \\
  \psi_+(0) \\
  \phi(0) \\
  \phi(\ell)
\end{pmatrix}.
\end{equation}

When the graph is exposed to a homogeneous magnetic field, the magnetic flux through the loop equals
$$
\Phi=\iint_{S}\vec{B}\cdot\d\vec{S}\,,
$$
where $S$ is the area of the loop (in case of a circle, $S=\ell^2/(4\pi)$).
The flux $\Phi$ can be expressed in terms of the magnetic vector potential $\vec{A}$,
$$
\Phi=\oint_{\partial S}\vec{A}\cdot\d\vec{l}\,.
$$
Therefore, if the magnetic field is perpendicular to the graph plane,
the strength of the vector potential on the loop is given as
\begin{equation}\label{A}
A=\frac{S}{\ell}\cdot B\,.
\end{equation}
The corresponding wave function component on the loop takes the form
\begin{equation}\label{C phi}
\phi(x)=C^+\e^{\i(\frac{q}{\hbar}A+k)x}+C^-\e^{\i(\frac{q}{\hbar}A-k)x} \quad\text{for all $x\in(0,\ell)$}\,.
\end{equation}
In order to find the transmission amplitude along the input--output line, we determine the Dirichlet-to-Neumann function. The condition $\phi(0)=\phi(\ell)=1$ means
$$
C^++C^-=C^+\e^{\i(\frac{q}{\hbar}A+k)\ell}+C^-\e^{\i(\frac{q}{\hbar}A-k)\ell}=1\,,
$$
which leads to
\begin{align*}
&
C^+=-\frac{1}{\sin k\ell}\e^{-\i\frac{(A+k)\ell}{2}}\sin\frac{(A-k)\ell}{2}\,, \\
&
C^-=\frac{1}{\sin k\ell}\e^{\i\frac{(-A+k)\ell}{2}}\sin\frac{(A+k)\ell}{2}\,.
\end{align*}
We substitute from here into expression~\eqref{C phi} and calculate $\Lambda(E)=\phi'(0_+)-\phi'(\ell_-)$. We obtain
\begin{align*}
\Lambda(E)&=-\frac{4k}{\sin k\ell}\sin\frac{(\frac{q}{\hbar}A-k)\ell}{2}\sin\frac{(\frac{q}{\hbar}A+k)\ell}{2} \\
&
=-2k\frac{\cos k\ell-\cos\frac{qBS}{\hbar}}{\sin k\ell}\,,
\end{align*}
where equation~\eqref{A} has been used to express the magnetic potential $A$ in terms of the magnetic field $B$.

\begin{observation}\label{ObsLambda}
There holds
\begin{itemize}
\item[(i)] $\lambda\in\sigma(H_{\Gamma})\Rightarrow\phi(0)\neq0$. Consequently, $\sigma_0=\emptyset$.
\item[(ii)] $\Lambda(\frac{\hbar^2k^2}{2m})$ as a function of $k$ has period $2\pi/\ell$.
\item[(iii)] If $E$ is fixed, then $\Lambda(E)$ as a function of $B$ has period $\frac{2\pi\hbar}{qS}$.
\item[(iv)] The equation $\Lambda(\frac{\hbar^2k^2}{2m})=0$ has exactly one solution in each of the intervals $[N\pi/\ell,(N+1)\pi/\ell]$, $N\in\N_0$, namely
\begin{equation}\label{peaks}
k_N=\left\{\begin{array}{cl}
\frac{q\{B\}S}{\hbar\ell}+\frac{N\pi}{\ell} & \text{if $N$ is even}\,, \\
-\frac{q\{B\}S}{\hbar\ell}+\frac{N\pi}{\ell} & \text{if $N$ is odd}\,,
\end{array}\right.
\end{equation}
where $\{B\}=B-\left\lfloor\frac{qBS}{2\pi\hbar}\right\rfloor\cdot\frac{2\pi\hbar}{qS}$.
\end{itemize}
\end{observation}
With respect to Observation~\ref{ObsLambda} (i), formula~\eqref{T} is applicable to every $k>0$:
$$
\mathcal{T}(k)=\left(1+\i\alpha^2\frac{\cos k\ell-\cos\frac{qBS}{\hbar}}{\sin k\ell}\right)^{-1} \qquad\text{for all $k>0$}\,.
$$
The transmission probabilitity is given as $\P^{(B)}(k)=|\mathcal{T}(k)|^2$, where the notation $\P^{(B)}(k)$ is used to emphasize its dependence on $B$. We have
\begin{equation}\label{P}
\P^{(B)}(k)
=\left[1+\alpha^4\left(\frac{\cos k\ell-\cos\frac{qBS}{\hbar}}{\sin k\ell}\right)^2\right]^{-1}\,.
\end{equation}

If $\alpha\gg1$ (strictly speaking, if $4\alpha^4\gg1$), then $\P^{(B)}(k)$ has sharp peaks attaining $1$ at the points $k_N$ given by~\eqref{peaks}, cf. Corollary~\ref{T coro}. The situation is illustrated in Figure~\ref{Fig: P}.
\begin{figure}[h]
\begin{center}
\includegraphics[width=6.5cm]{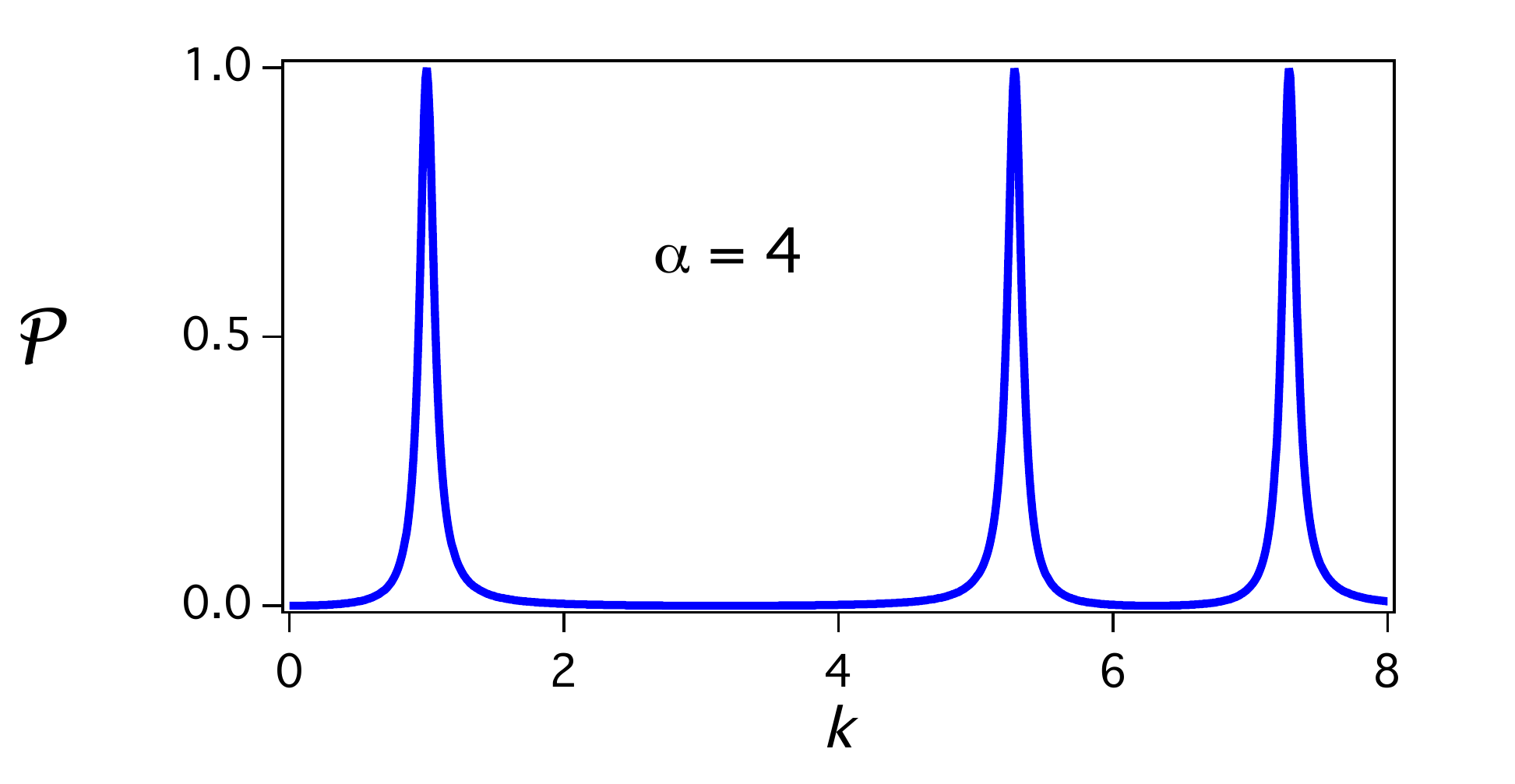}
\caption{The transmission characteristics of the graph depicted in Figure~\ref{Fig: loop} with parameter $\alpha=4$. The transmission probability is plotted for $B=\frac{\hbar}{qS}$. The scale of the particle wavenumber $k$ is chosen such that $k=1$ corresponds to $k=1/\ell$.}
\label{Fig: P}
\end{center}
\end{figure}
Since the positions of the peaks depend on the magnetic field $B$ by a quite simple relation, cf. equation~\eqref{peaks}, the graph can be used for controllable spectral filtering.
Let us assume that the wavenumbers of particles coming in the vertex be in the interval $[0,k_{\max}]$, where $k_{\max}=\pi/\ell$.
Let us define
\begin{equation*}
B_{\max}=\frac{\pi\hbar}{qS}\,.
\end{equation*}
According to Observation~\ref{ObsLambda} (iv), the function $\P^{(B)}(k)$ has a single peak attaining $1$ in the interval $[0,k_{\max}]$. For any $B\in\left[0,B_{\max}\right]$, the peak is located at $k=k_0=\frac{qBS}{\hbar\ell}$.
In other words, if $B$ raises from $0$ to $B_{\max}=\frac{\pi\hbar}{qS}$, the position of the peak of $\P^{(B)}(k)$ shifts from $0$ to $k_{\max}=\pi/\ell$, cf. Figure~\ref{Fig: peak}.
\begin{figure}[h]
\begin{center}
\includegraphics[width=6.5cm]{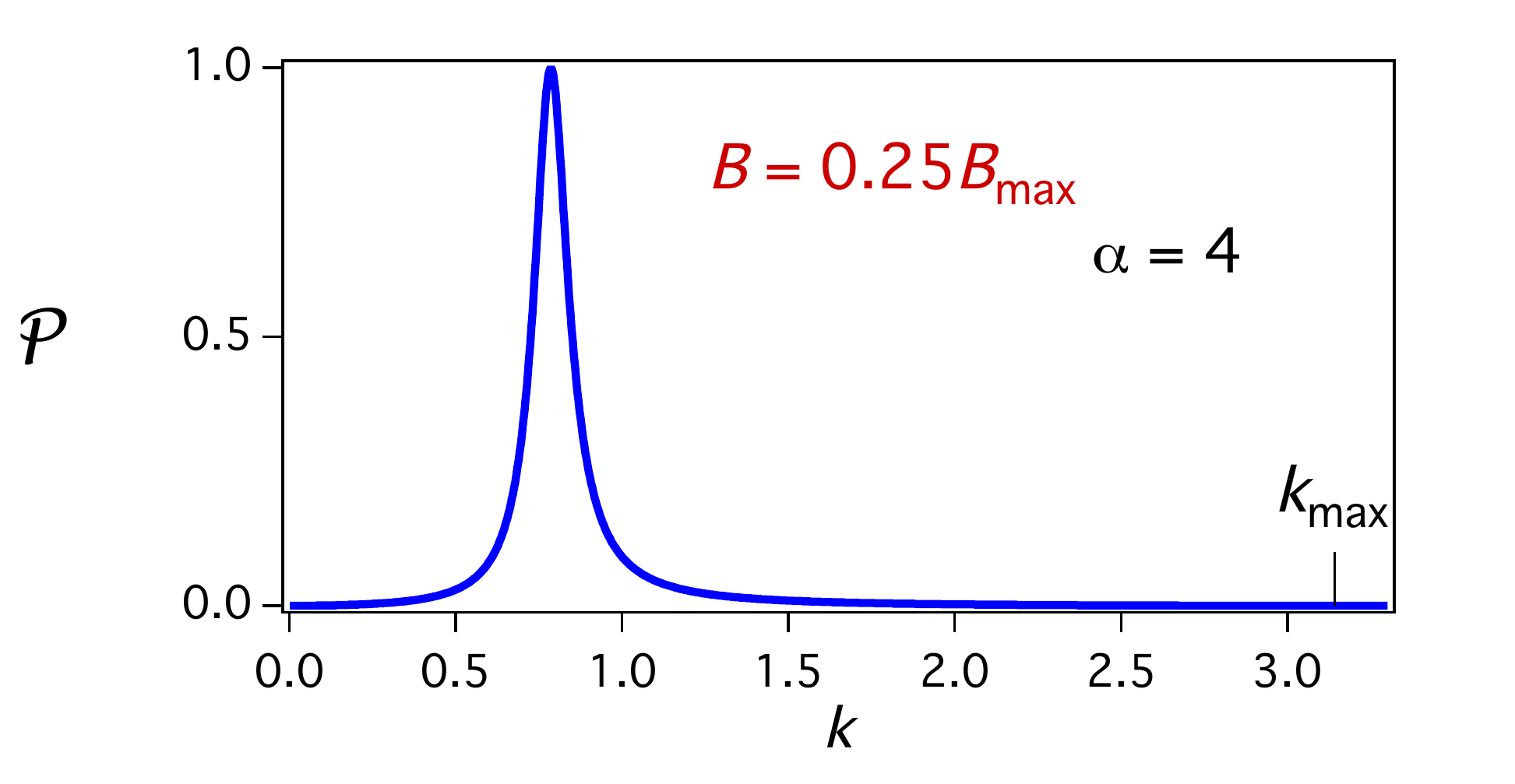} \\
\includegraphics[width=6.5cm]{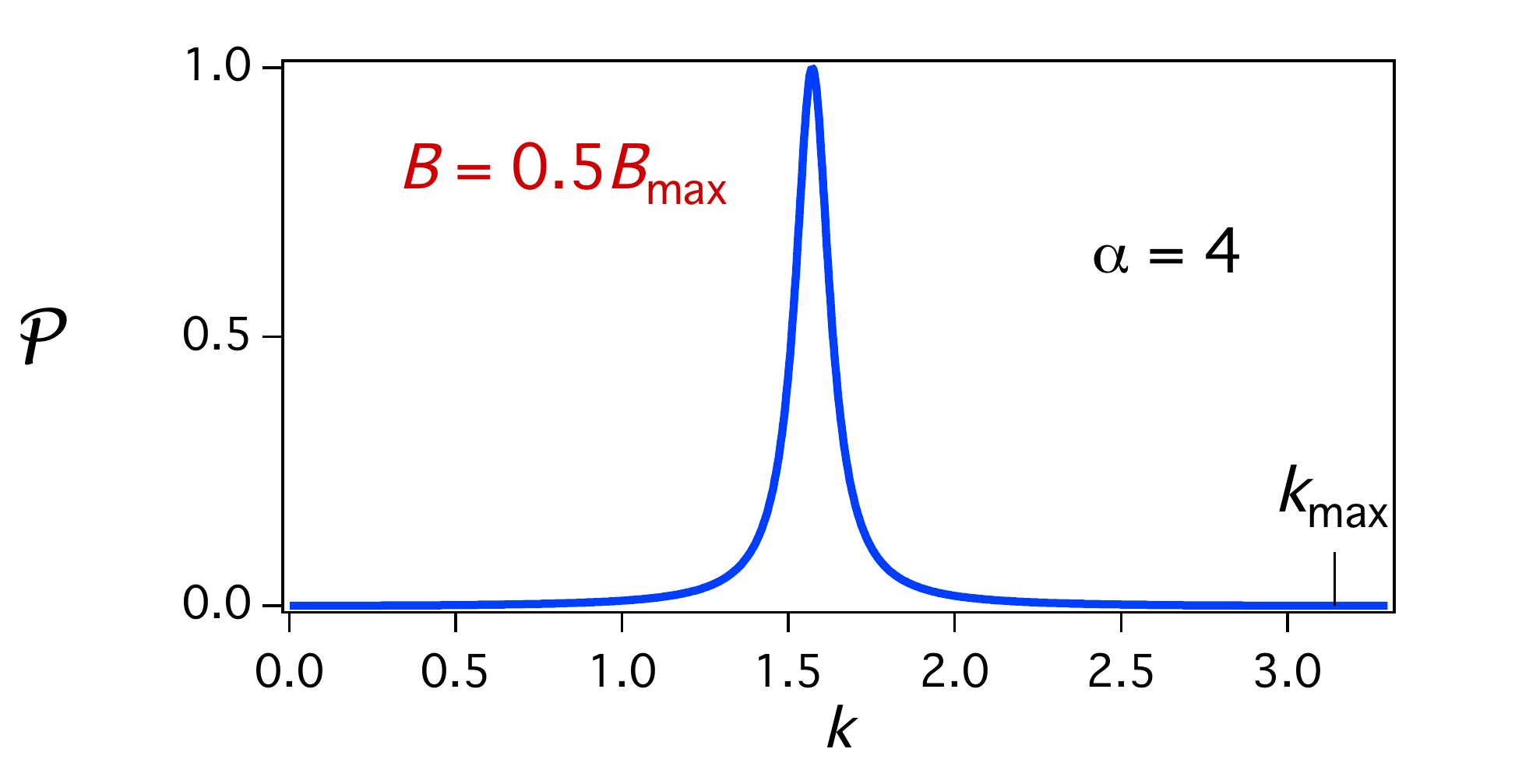}
\caption{The graph depicted in Figure~\ref{Fig: loop} can be used as a spectral filter controllable by a magnetic field. If the magnetic field $B$ ranges over the interval $[0,B_{\max}]$ for $B_{\max}=\frac{\pi\hbar}{qS}$, the passband position ranges over $\left[0,k_{\max}\right]$, where $k_{\max}=\pi/\ell$.}
\label{Fig: peak}
\end{center}
\end{figure}


\section{Physical realization of scale-invariant vertex}\label{Section: Approximation}

Despite the scale-invariant couplings represent nontrivial point interactions with no straightforward physical interpretation, it is known that they can be approximately constructed using several $\delta$ potentials~\cite{CET10,CT10,TC12}.
In this section we demonstrate how to produce the coupling given by boundary conditions~\eqref{bc}. The solution will be obtained by applying the technique from the paper~\cite{TC12}.

The procedure
begins with transforming the scale-invariant boundary conditions in the given vertex of degree $N$ into their $ST$-form,
$$
\left(\begin{array}{cc}
I^{(r)} & T \\
0 & 0
\end{array}\right)\Psi'(0)=
\left(\begin{array}{cc}
0 & 0 \\
-T^* & I^{(N-r)}
\end{array}\right)\Psi(0)
$$
(cf.~\eqref{ST}, recall that $S=0$).
Note that the boundary conditions we consider~\eqref{bc} are already in this form,
and $N=2+n$,
$r=1$, $T=\begin{pmatrix} 1 & \alpha & \cdots & \alpha \end{pmatrix}$. 

In the next step we take $N$ decoupled lines, and for every $i=1,\ldots,r$, we connect the endpoint $v_i$ of the line numbered by $i$ with the endpoints $v_j$ of certain other lines by short lines (``auxiliary links''). In our case we have $r=1$, thus the endpoint $v_1$ shall be connected with (certain of) the endpoints $v_2,v_3,\ldots,v_N$ (Figure~\ref{Fig: approximation}). If the entries of $T$ are indexed in the way $T=\begin{pmatrix} t_{12} & t_{13} & \cdots & t_{1N} \end{pmatrix}$, the connections are constructed according to the following criterion:
\emph{The endpoint $v_1$ is connected with $v_j$ for $j\in\{2,3,\ldots,2+n\}$ by a link if and only if $t_{1j}\neq0$.}
The links between $v_1$ and $v_j$ have the lengths $d/|t_{1j}|$, where $d$ is a common length parameter. The value of $d$ shall be chosen small enough, namely, $d\ll 1/k_{\max}$, where $k_{\max}$ is the maximal wavenumber of particles coming in the vertex.
If all the entries of $T$ are nonnegative, the links are chosen as just plain lines that do not carry any additional potentials~\cite{TC12}. 
\begin{figure}[h]
\begin{center}
\includegraphics[width=6.0cm]{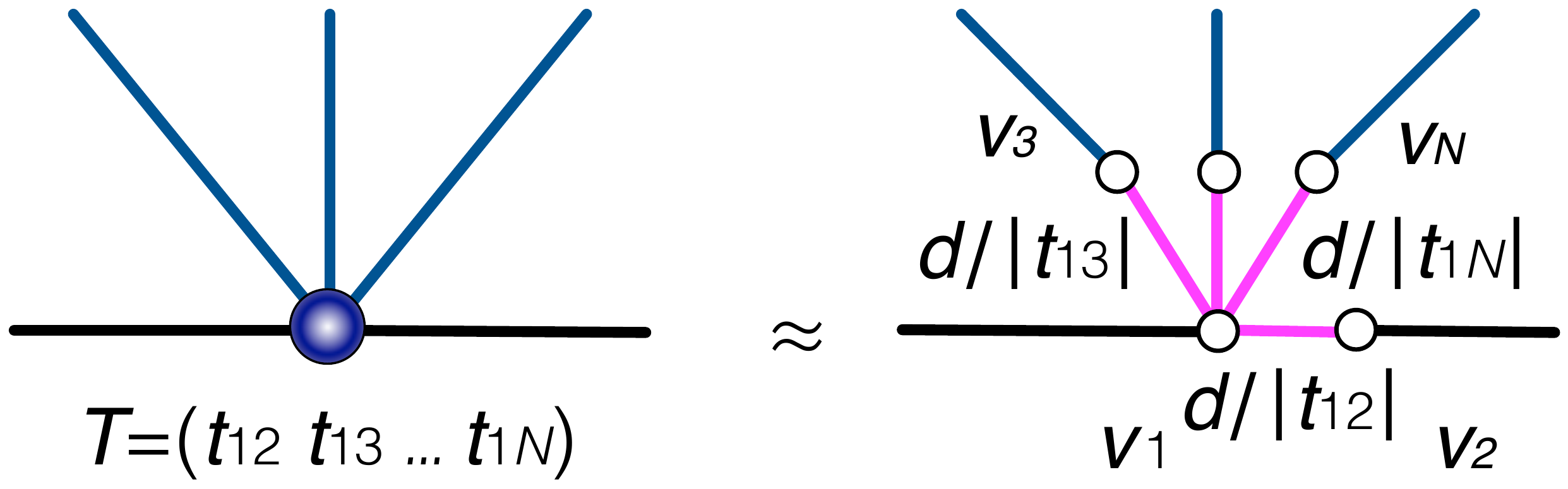}
\caption{Approximative construction of the vertex coupling given by boundary conditions~\eqref{bc}. The vertices $v_1,\ldots,v_N$ support $\delta$ potentials of properly chosen strengths. The length parameter $d$ governs the accuracy of the aproximation, the requested coupling is obtained in the limit $d\to0$.}
\label{Fig: approximation}
\end{center}
\end{figure}

In the endpoints $v_j$, $j=1,\ldots,N$, $\delta$-couplings are imposed. If $r=1$ and the entries of $T$ are all nonnegative, the strengths $\alpha_j$ of the $\delta$ potentials are given by the following formulae~\cite{TC12}:
\begin{align*}
\text{in the vertex $v_1$:}\qquad &\frac{1}{d}\left(\sum_{i=2}^N t_{1i}^2-\sum_{i=2}^N t_{1i}\right)\,;\\
\text{in the vertices $v_j$, $j=2,\ldots,N$:}\qquad &\frac{1}{d}\left(1-t_{1j}\right)\,.
\end{align*}

For $T=\begin{pmatrix} 1 & \alpha & \cdots & \alpha \end{pmatrix}$ we obtain the following arrangement.
\begin{itemize}
\item The link between $v_1$ and $v_2$ is of length $d$,
\item for every $j\in\{3,\ldots,N\}$, the link between $v_1$ and $v_j$ is of length $d/\alpha$.
\end{itemize}
The strength of the $\delta$-couplings in the endpoints $v_1,\ldots,v_N$ are
\begin{align*}
\text{in the vertex $v_1$:}\qquad &(N-2)\frac{\alpha(\alpha-1)}{d}\,;\\
\text{in the vertex $v_2$:}\qquad &0\,;\\
\text{in the vertices $v_j$ for $j=3,\ldots,N$:}\qquad &\frac{1-\alpha}{d}\,.
\end{align*}
Since there is no potential in the vertex $v_2$, the vertex has no importance any more and can be left out from the approximation arrangement, see Figure~\ref{Fig: implement}.
\begin{figure}[h]
\begin{center}
\includegraphics[width=6.0cm]{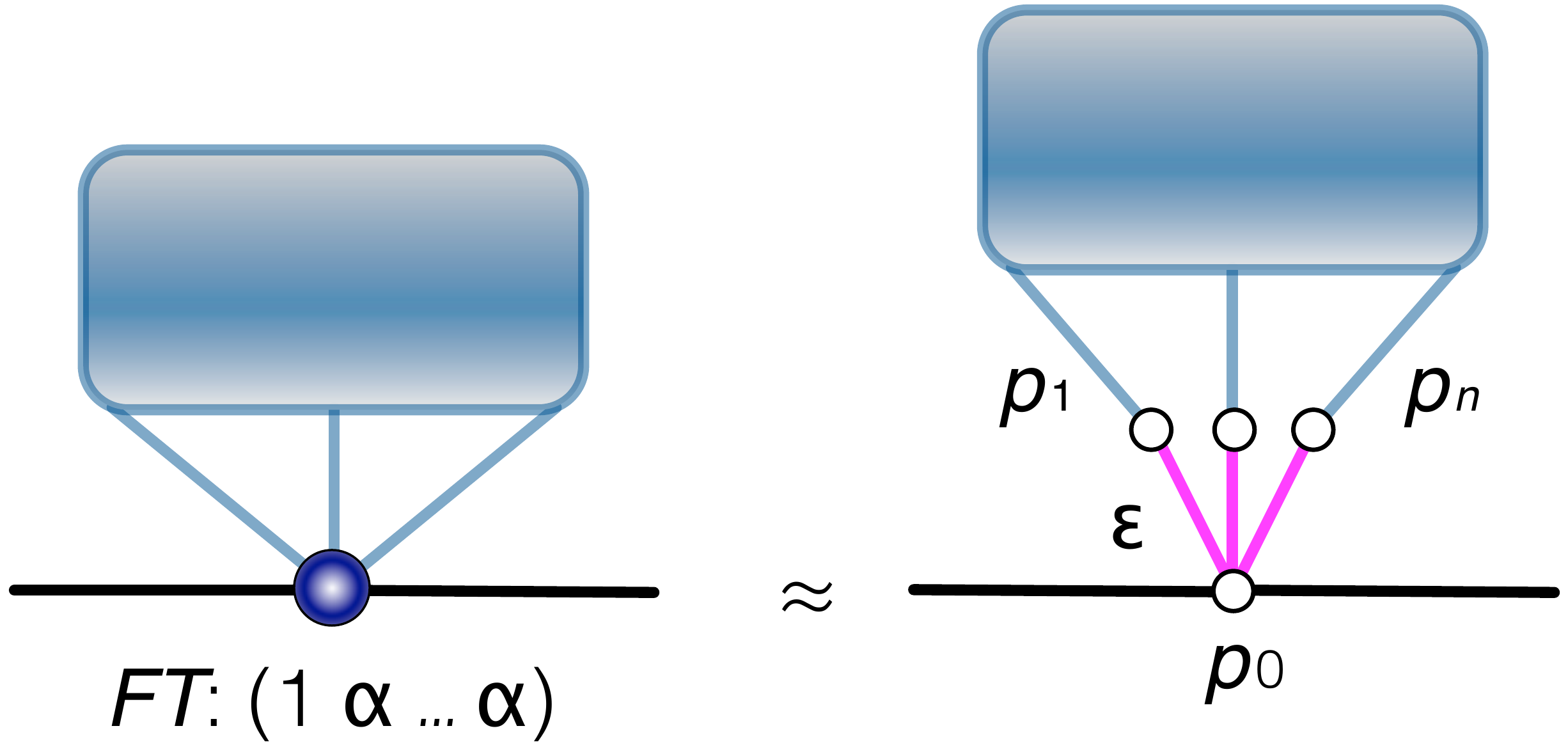}
\caption{A practical realization of the coupling~\eqref{bc} in the vertex $0$ of the graph $\Gamma_{\mathbf{io}}$. The $n$ edges of the graph $\Gamma$ that are incident to $v_0$ are decoupled and their endpoints are attached to a point on the input--output line by short edges of length $\epsilon$. Finally, the $\delta$-couplings, represented by small circles, are imposed in the endpoints.}
\label{Fig: implement}
\end{center}
\end{figure}
Leaving out the vertex $v_2$ allows us to introduce more simplifications.
\begin{itemize}
\item We denote the vertex $v_1$ by $p_0$ and vertices $v_3,\ldots,v_{n+2}$ by $p_1,\ldots,p_{n}$;
\item we introduce the length parameter $\epsilon\equiv d/\alpha$;
\item we denote the strengths of the $\delta$-couplings in vertices $p_j$ by $\alpha_j$ for all $j=0,1,\ldots,n$, i.e.,
\begin{equation}\label{alpha}
\alpha_0=n\frac{\alpha-1}{\epsilon}\,,\qquad \alpha_1=\cdots=\alpha_{n}=\frac{1-\alpha}{\alpha\epsilon}\,.
\end{equation}
\end{itemize}
The implementation of the approximation in the graph $\Gamma_{\mathbf{io}}$ is illustrated in Figure~\ref{Fig: implement}.
The small size limit $\epsilon\to 0$, together with the $\delta$~potentials strengths properly scaled according to the formulae~\eqref{alpha} above, effectively produces the required scale-invariant coupling in vertex $0$, as we show in Theorem~\ref{convergence} below.



\begin{theorem}\label{convergence}
Let $\mathcal{T}_\epsilon(k)$ be the transmission amplitude of the approximating graph depicted in Figure~\ref{Fig: implement}. Let $\sigma_0$ and $\Lambda$ have the same meaning as in Section~\ref{Section: Main}, defined by \eqref{sigma0} and \eqref{Lambda}.
\begin{itemize}
\item[(i)] If $E\in(0,+\infty)\backslash\sigma_0$, then
\begin{equation*}
\lim_{\epsilon\to0}\mathcal{T}_\epsilon(k)=\frac{1}{1+\alpha^2\frac{\Lambda(E)}{2\i k}}\,.
\end{equation*}
\item[(ii)] If $E\in\sigma_0$, then $\lim_{\epsilon\to0}\mathcal{T}_\epsilon(k)=0$.
\end{itemize}
\end{theorem}

\begin{proof}
We denote the wave function components on the auxiliary short lines by $\varphi_j(x)$ for $x\in[0,\epsilon]$, where $x=0$ corresponds to the point $p_0$ and $x=\epsilon$ to the point $p_j$. Therefore, these components take the form
\begin{equation}\label{varphi_j}
\varphi_j(x)=C_j^+\e^{\i kx}+C_j^-\e^{-\i kx}\,.
\end{equation}
The $\delta$-coupling in the vertex $p_0$ is expressed by the boundary conditions
\begin{gather}
\psi_-(0)=\psi_+(0)=\varphi_j(0) \qquad\text{for all $j=1,\ldots,n$}\,, \label{psi a}\\
-\psi_-'(0)+\psi_+'(0)+\sum_{j=1}^{n}\varphi_j'(0)=\alpha_0\psi_-(0)\,, \qquad\text{where $\alpha_0=n\frac{\alpha-1}{\epsilon}$}\,. \label{psi b}
\end{gather}
The $\delta$-interactions in the points $p_j$ mean
\begin{gather}
\varphi_j\left(\epsilon\right)=\phi_j(0) \qquad\text{for all $j=1,\ldots,n$}\,, \label{varphi a}\\
\phi_j'(0)+\varphi_j'\left(\epsilon\right)=\alpha_j\phi_j(0)\,, \qquad\text{where $\alpha_j=\frac{1-\alpha}{\alpha\epsilon}$}\,. \label{varphi b}
\end{gather}
Our first goal is to express $\varphi_j(0)$ and $\varphi_j'(0)$ in terms of $\phi_j(0)$ and $\phi_j'(0)$. We start by substituting formula~\eqref{varphi_j} into the boundary conditions \eqref{varphi a} and \eqref{varphi b}, which leads to the system
\begin{gather*}
C_j^+\e^{\i k\epsilon}+C_j^-\e^{-\i k\epsilon}=\phi_j(0) \qquad\text{for all $j=1,\ldots,n$}\,, \\
\phi_j'(0)+\i k\left(C_j^+\e^{\i k\epsilon}+C_j^-\e^{-\i k\epsilon}\right)=\frac{1-\alpha}{\alpha\epsilon}\phi_j(0)\,.
\end{gather*}
Hence we obtain $C_j^+$ and $C_j^-$;
\begin{subequations}\label{C_j}
\begin{align}
C_j^+&=\frac{1}{2\i k}\left[\phi_j'(0)+\left(\i k-\frac{1-\alpha}{\alpha\epsilon}\right)\right]\,, \\
C_j^-&=\frac{1}{2\i k}\left[-\phi_j'(0)+\left(\i k+\frac{1-\alpha}{\alpha\epsilon}\right)\right]\,.
\end{align}
\end{subequations}
Equation~\eqref{varphi_j} imply
\begin{equation}\label{varphi(0)}
\varphi_j(0)=C_j^++C_j^- \qquad\text{and}\qquad \varphi_j'(0)=\i k(C_j^+-C_j^-)\,.
\end{equation}
We substitute the expressions~\eqref{C_j} into equations~\eqref{varphi(0)}, which gives
\begin{align}
\varphi_j(0)&=-\frac{\sin k\epsilon}{k}\phi_j'(0)+\left(\cos k\epsilon+\frac{1-\alpha}{\alpha\epsilon k}\sin k\epsilon\right)\phi_j(0)\,, \label{varphi_j(0)} \\
\varphi_j'(0)&=\frac{\cos k\epsilon}{\i k}\phi_j'(0)+\left(-\i\sin k\epsilon+\i\frac{1-\alpha}{\alpha\epsilon k}\cos k\epsilon\right)\phi_j(0)\,. \label{varphi_j'(0)}
\end{align}

Having expressed $\varphi_j(0)$ and $\varphi_j'(0)$, we use formulas \eqref{varphi_j(0)} and \eqref{varphi_j'(0)} in boundary conditions \eqref{psi a} and \eqref{psi b}. Hence we get the system
\begin{equation}\label{bc apr1}
\psi_-(0)=\psi_+(0)=-\frac{\sin k\epsilon}{k}\phi_j'(0)+\left(\cos k\epsilon+\frac{1-\alpha}{\alpha\epsilon k}\sin k\epsilon\right)\phi_j(0) \qquad\text{for all $j=1,\ldots,n$}
\end{equation}
and
\begin{equation}\label{bc apr2}
-\psi_-'(0)+\psi_+(0)+\sum_{j=1}{n}\left[\cos k\epsilon\cdot\phi_j'(0)+\left(k\sin k\epsilon-\frac{1-\alpha}{\alpha\epsilon}\cos k\epsilon\right)\phi_j(0)\right]=n\frac{\alpha-1}{\epsilon}\psi_-(0)\,.
\end{equation}

From now on we proceed similarly as in the proof of Proposition~\ref{Prop. T}.

(i)\quad If $E\in(0,+\infty)\backslash\sigma_0$, the problem~\eqref{Problem} has a solution $\Phi=(\tilde{\phi}_1,\ldots,\tilde{\phi}_n)^T$. We set
\begin{gather*}
\psi_-(x)=\e^{\i kx}+\mathcal{R}_\epsilon(k)\e^{-\i kx}\,,\qquad \psi_+(x)=\mathcal{T}_\epsilon(k)\e^{\i kx}\,, \\
\phi_j=c\tilde{\phi}_j \quad\text{for all $j=1,\ldots,n$}\,,
\end{gather*}
where $c\in\C$ is a constant to be specified later,
and
$$
\varphi_j=C_j^+\e^{\i kx}+C_j^-\e^{-\i kx} \quad\text{for $C_j^+,C_j^-$ given by equations~\eqref{C_j}}\,.
$$
These functions obey the boundary conditions in each $v\in V_\Gamma\backslash\{v_0\}$ due to the assumptions~\eqref{Problem} (recall that $\Gamma$ denotes the ``attached'' graph), as well as the boundary conditions corresponding to the $\delta$-interactions in the vertices $p_1,\ldots,p_{n}$. Therefore, we require these functions to satisfy also the boundary conditions \eqref{bc apr1} and \eqref{bc apr2} in the vertex $p_0$, whence we find $\mathcal{T}_\epsilon(k)$. Using the properties of $\Phi$ (cf.~\eqref{Problem}), we can rewrite the boundary conditions \eqref{bc apr1} and \eqref{bc apr2} as
\begin{gather*}
1+\mathcal{R}_\epsilon(k)=\mathcal{T}_\epsilon(k)=-\frac{\sin k\epsilon}{k}c\tilde{\phi}_j'(0)+\left(\cos k\epsilon+\frac{1-\alpha}{\alpha\epsilon k}\sin k\epsilon\right)c\,; \\
-\i k(1-\mathcal{R}_\epsilon(k))+\i k\mathcal{T}_\epsilon(k)+\cos k\epsilon \cdot c\Lambda(E)+n\left(k\sin k\epsilon-\frac{1-\alpha}{\alpha\epsilon}\cos k\epsilon\right)c=n\frac{\alpha-1}{\epsilon}(1+\mathcal{R}_\epsilon(k))\,.
\end{gather*}
This system yields
$$
\mathcal{T}_\epsilon(k)=\frac{2\i k\left[-\frac{\sin k\epsilon}{k}\frac{\Lambda(E)}{n}+\left(\cos k\epsilon+\frac{1-\alpha}{\alpha\epsilon k}\sin k\epsilon\right)\right]}{\cos k\epsilon\cdot\Lambda(E)+n\left(k\sin k\epsilon-\frac{1-\alpha}{\alpha\epsilon}\cos k\epsilon\right)+\left(2\i k-n\frac{\alpha-1}{\epsilon}\right)\left[-\frac{\sin k\epsilon}{k}\frac{\Lambda(E)}{n}+\left(\cos k\epsilon+\frac{1-\alpha}{\alpha\epsilon k}\sin k\epsilon\right)\right]}\,.
$$
A simple calculation gives
$$
\lim_{\epsilon\to0}\mathcal{T}_\epsilon(k)=\frac{2\i k}{2\i k+\alpha^2\Lambda(E)}\,,
$$
i.e., $\lim_{\epsilon\to0}\mathcal{T}_\epsilon(k)=\mathcal{T}(k)$ for $\mathcal{T}(k)$ given by \eqref{T}.

(ii)\quad If $E\in\sigma_0$, let $\Psi_\epsilon=(\e^{\i kx}+\mathcal{R}_\epsilon(k)\e^{-\i kx},\mathcal{T}_\epsilon(k)\e^{\i kx},\phi_1,\ldots,\phi_{n},\varphi_1,\ldots,\varphi_{n})^T$ be the final-state wave function on the approximating graph. Then $\Phi_\epsilon:=\left(\phi_1,\ldots,\phi_{n}\right)^T$ obviously satisfies~\eqref{Problem 1} and \eqref{Problem 2} for every $\epsilon$. Moreover, $\Psi_\epsilon$ obeys the boundary conditions~\eqref{bc apr1}, hence
\begin{equation}\label{lim ii}
\lim_{\epsilon\to0}\mathcal{T}_\epsilon(k)=\lim_{\epsilon\to0}\phi_1=\cdots=\lim_{\epsilon\to0}\phi_{n}\,.
\end{equation}
To sum up, the function $\lim_{\epsilon\to0}\Phi_\epsilon=(\tilde{\phi}_1,\ldots,\tilde{\phi}_{n})^T$ satisfies the conditions \eqref{Problem 1}, \eqref{Problem 2} and $\tilde{\phi}_j=c$ for a certain value $c$. However, since we assume $E\in\sigma_0$, the problem~\eqref{Problem} cannot have a solution, which implies $c=0$. Therefore, $\lim_{\epsilon\to0}\mathcal{T}_\epsilon(k)=0$ according to~\eqref{lim ii}.

\end{proof}

Theorem~\ref{convergence} means that the transmission amplitude $\mathcal{T}_\epsilon(k)$ on the approximating graph, sketched in Figure~\ref{Fig: implement}, satisfies
$\lim_{\epsilon\to0}\mathcal{T}_\epsilon(k)=\mathcal{T}(k)$ for all $k>0$.

\begin{remark}
Besides the procedure described above there exists an alternative approximate construction of exotic graph vertices, which is based on the use of tubular networks built over the graph~\cite{EP12}.
\end{remark}

\section{Band-stop filter}\label{Section: Inverse}

The idea, used in Section~\ref{Section: Main} for the construction of a band-pass spectral filter, can be extended. In this section we apply it to design a band-stop filter. It is built upon the same graph $\Gamma_{\mathbf{io}}=(\{0\}\cup V_\Gamma\backslash\{v_0\},E_\Gamma\cup\{\mathbf{i},\mathbf{o}\})$ as before, but the scale-invariant coupling~\eqref{bc} in the vertex $0$ is replaced by
\begin{widetext}
\begin{equation}\label{bc inv}
\begin{pmatrix}
1 & 0 & \alpha & \cdots & \alpha \\
0 & 1 & \alpha & \cdots & \alpha \\
0 & 0 & 0 & \cdots & 0 \\
\vdots & \vdots & \vdots &  & \vdots \\
0 & 0 & 0 & \cdots & 0
\end{pmatrix}
\begin{pmatrix}
-\psi_-'(0) \\
\psi_+'(0) \\
\phi_1'(0) \\
\vdots \\
\phi_{n}'(0)
\end{pmatrix}
=
\begin{pmatrix}
0 & 0 & 0 & \cdots & 0 \\
0 & 0 & 0 & \cdots & 0 \\
-\alpha & -\alpha & 1 & \cdots & 0 \\
\vdots & \vdots &  & \ddots &  \\
-\alpha & -\alpha & 0 & \cdots & 1
\end{pmatrix}
\begin{pmatrix}
\psi_-(0) \\
\psi_+(0) \\
\phi_1(0) \\
\vdots \\
\phi_{n}(0)
\end{pmatrix}\,,
\end{equation}
\end{widetext}
where $\alpha>0$ is a parameter.

Note that the boundary conditions~\eqref{bc inv} are related to the boundary conditions~\eqref{bc} by the duality relations
\begin{equation}\label{transformation}
\begin{array}{ccc}
\text{b. c. \eqref{bc}} & & \text{b. c. \eqref{bc inv}}
\\ \hline
-\psi_-'(0)+\psi_+'(0) & \quad\longleftrightarrow\quad & -\psi_-'(0)=\psi_+'(0)\,, \\
\psi_-(0)=\psi_+(0) & \quad\longleftrightarrow\quad & \psi_-(0)+\psi_+(0)\,,
\end{array}
\end{equation}
which is a generalisation of the $\delta$-$\delta^\prime$ duality found earlier~\cite{CS99,CFT01}.
If we express the transformation relations~\eqref{transformation} in terms of the scattering amplitudes $\mathcal{R}(k)$ and $\mathcal{T}(k)$ using expressions~\eqref{psi_io}, we obtain the system of equations
\begin{equation}\label{system duality}
\begin{array}{ccc}
\i k(-1+\mathcal{R}_\mathrm{pass}(k)+\mathcal{T}_\mathrm{pass}(k)) & \quad=\quad & \i kC(-1+\mathcal{R}_\mathrm{stop}(k))=\i kC\mathcal{T}_\mathrm{stop}(k)\,, \\
1+\mathcal{R}_\mathrm{pass}(k)=\mathcal{T}_\mathrm{pass}(k) & \quad=\quad & C(1+\mathcal{R}_\mathrm{stop}(k)+\mathcal{T}_\mathrm{stop}(k))\,,
\end{array}
\end{equation}
where $\mathcal{T}_\mathrm{pass}(k)$ is the transmission amplitude of the band-pass filter, derived in Proposition~\ref{Prop. T}, $\mathcal{T}_\mathrm{stop}(k)$ is the transmission amplitude of the band-stop filter constructed for the boundary conditions~\eqref{bc inv}, and $C$ is a coefficient which is needed to harmonize the normalizations.
Solving the system~\eqref{system duality}, we get a formula relating 
the transmission amplitude of the band-pass filter from Section~\ref{Section: Main} with the transmission amplitude of the band-stop filter introduced in this section,
\begin{equation}\label{duality}
2\left(1-\frac{1}{\mathcal{T}_\mathrm{pass}(k)}\right) = \frac{1}{2\left(1+\frac{1}{\mathcal{T}_\mathrm{stop}(k)}\right)}\,.
\end{equation}

Recall that $\lim_{\alpha\to0}\mathcal{T}_\mathrm{pass}(k)\in\{0,1\}$ for all $k>0$, cf. Proposition~\ref{Prop. T}. The duality relation~\eqref{duality} then implies
\begin{align*}
\lim_{\alpha\to0}\mathcal{T}_\mathrm{pass}(k)=1 \quad&\Rightarrow\quad \lim_{\alpha\to0}\mathcal{T}_\mathrm{stop}(k)=0\,; \\
\lim_{\alpha\to0}\mathcal{T}_\mathrm{pass}(k)=0 \quad&\Rightarrow\quad \lim_{\alpha\to0}\mathcal{T}_\mathrm{stop}(k)=-1\,. \\
\end{align*}
This observation together with Corollary~\ref{T coro} leads to the following result:
\begin{proposition}
For every $k>0$, the transmission probability of the graph with the boundary conditions~\eqref{bc inv} in the vertex $0$ obeys
$$
\lim_{\alpha\to\infty}\P(k)=
\left\{\begin{array}{cl}
0 & \text{if $E=\frac{\hbar^2k^2}{2m}\in\sigma(H_{\Gamma})\backslash\sigma_0$}\,, \\
1 & \text{otherwise}\,.
\end{array}\right.
$$
\end{proposition}
Consequently, the graph studied in Section~\ref{Section: Main} and the graph studied in this section have complementary transmission characteristics. The device based on the scale-invariant vertex coupling~\eqref{bc inv} works as a band-stop filter with stopbands at $E\in\sigma(H_{\Gamma})\backslash\sigma_0$.

The duality relation~\eqref{duality} together with the result of Proposition~\ref{Prop. T} also allows one to easily find the formula for the transmission amplitude $\mathcal{T}_\mathrm{stop}(k)$ on the graph $\Gamma_{\mathbf{io}}$ with boundary conditions~\eqref{bc inv} in the vertex $0$.

\begin{proposition}\label{Prop. T inv}
Let $\mathcal{T}(k)$ be the amplitude of the transmission to the output line for the boundary conditions~\eqref{bc inv} and for an incoming particle of energy $E=\frac{\hbar^2k^2}{2m}$. Let $\sigma_0$ and $\Lambda$ have the same meaning as in Section~\ref{Section: Main}, cf.~\eqref{sigma0} and~\eqref{Lambda}.
It holds:
\begin{itemize}
\item[(i)] If $E=\frac{\hbar^2k^2}{2m}\in(0,+\infty)\backslash\sigma_0$, then
\begin{equation}\label{T inv}
\mathcal{T}(k)=-\frac{\Lambda(E)}{\Lambda(E)+\frac{\i k}{2\alpha^2}}\,.
\end{equation}
\item[(ii)] If $E=\frac{\hbar^2k^2}{2m}\in\sigma_0$, then $\mathcal{T}(k)=-1$.
\end{itemize}
\end{proposition}

\section{Combined filtering}\label{Section: Switch}

In this section we construct a spectral filtering device with two outputs that combines the properties of the generic band-pass spectral filter studied in Section~\ref{Section: Main} and of the band-stop filter described in Section~\ref{Section: Inverse}. The device can serve as a quantum spectral separator, or a switch.

\begin{figure}[h]
\begin{center}
\includegraphics[width=4.2cm]{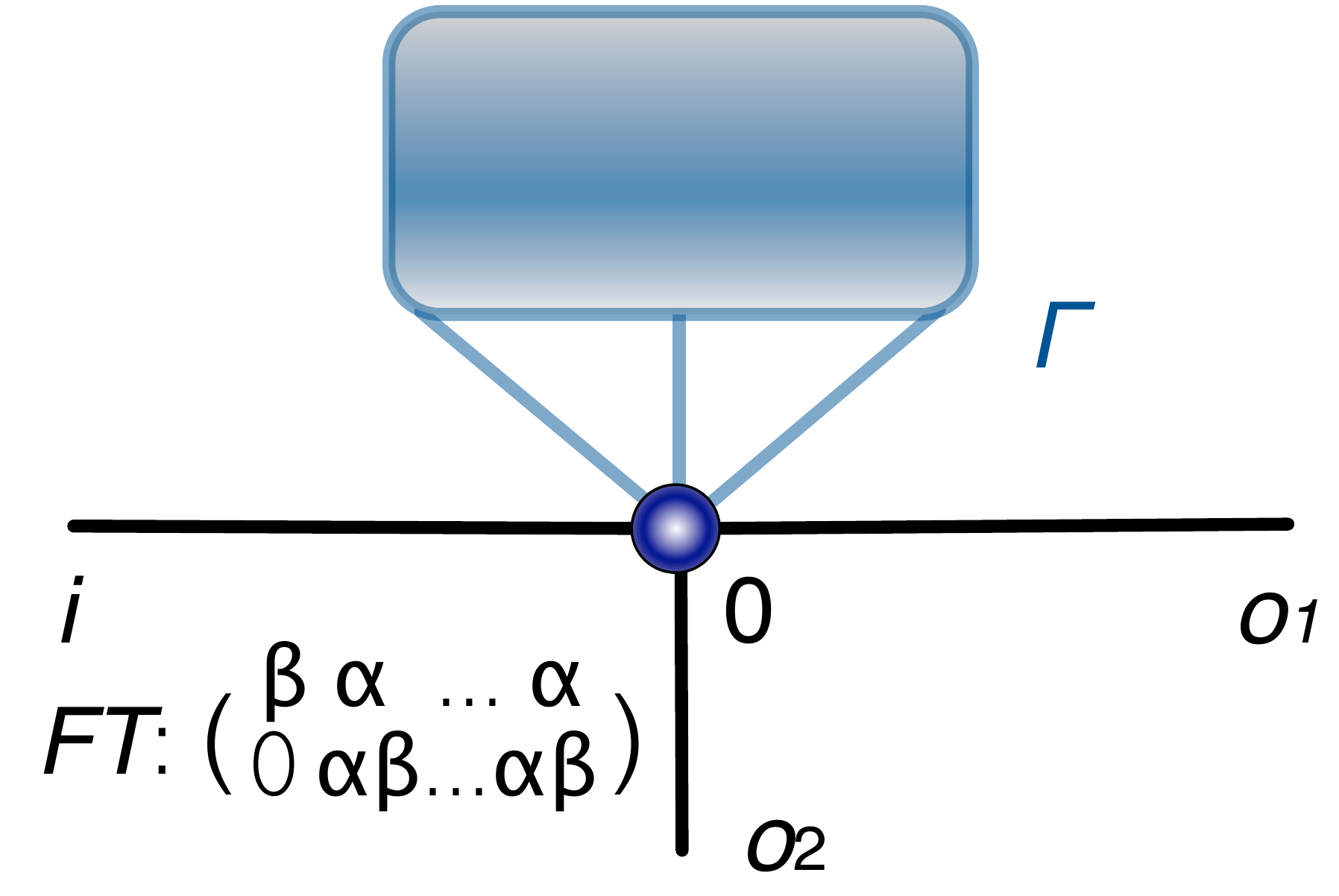}
\caption{A schematic illustration of the quantum spectral separator.}
\label{Fig: switch}
\end{center}
\end{figure}
Let $\Gamma$ be a quantum graph having the properties introduced in section~\ref{Section: Main}. In particular, there exists a vertex $v_0\in V_\Gamma$ with free boundary conditions. Let $\Gamma_{\mathbf{ioo}}$ be the quantum graph constructed from the graph $\Gamma$ by attaching an input half line $\mathbf{i}$ and two output half lines $\mathbf{o1},\mathbf{o2}$ to the vertex $v_0$, see Figure~\ref{Fig: switch}.
We denote the vertex incident to $\mathbf{i,o1,o2}$ by $0$, similarly as in Section~\ref{Section: Main}, hence $\Gamma_{\mathbf{ioo}}=(\{0\}\cup V_\Gamma\backslash\{v_0\},E_\Gamma\cup\{\mathbf{i},\mathbf{o}_1,\mathbf{o}_2\})$. The wave function components on the input half-line and on the output half lines will be denoted by $\psi_-$ and $\psi_1,\psi_2$, respectively.
The filtering function of the graph $\Gamma_{\mathbf{ioo}}$ relies on a scale-invariant coupling in the vertex $0$, described by the following boundary conditions:
\begin{widetext}
\begin{equation}\label{bc switch}
\begin{pmatrix}
1 & 0 & \beta & \alpha & \cdots & \alpha \\
0 & 1 & 0 & \alpha\beta & \cdots & \alpha\beta \\
0 & 0 & 0 & 0 & \cdots & 0 \\
0 & 0 & 0 & 0 & \cdots & 0 \\
\vdots &  &  &  &  & \vdots \\
0 & 0 & 0 & 0 & \cdots & 0
\end{pmatrix}
\begin{pmatrix}
-\psi_-'(0) \\
\psi_1'(0) \\
\psi_2'(0) \\
\phi_1'(0) \\
\vdots \\
\phi_{n}'(0)
\end{pmatrix}
=
\begin{pmatrix}
0 & 0 & 0 & 0 & \cdots & 0 \\
0 & 0 & 0 & 0 & \cdots & 0 \\
-\beta & 0 & 1 & 0 & \cdots & 0 \\
-\alpha & -\alpha\beta & 0 & 1 & \cdots & 0 \\
\vdots & \vdots & \vdots &  & \ddots &  \\
-\alpha & -\alpha\beta & 0 & 0 & \cdots & 1
\end{pmatrix}
\begin{pmatrix}
\psi_-(0) \\
\psi_1(0) \\
\psi_2(0) \\
\phi_1(0) \\
\vdots \\
\phi_{n}(0)
\end{pmatrix}\,.
\end{equation}
\end{widetext}
Values $\alpha>0$ and $\beta>0$ are parameters of the coupling.

The wave function component $\psi_-$ is a superposition of the incoming and the reflected wave, and the components $\psi_1$ and $\psi_2$ represent outgoing waves, hence
\begin{subequations}\label{psi_-12}
\begin{align}
\psi_-(x)&=\e^{ikx}+\mathcal{R}(k)\e^{-ikx}\,, \label{psi_-} \\
\psi_1(x)&=\mathcal{T}_1(k)\e^{ikx}\,, \label{psi_1} \\
\psi_2(x)&=\mathcal{T}_2(k)\e^{ikx}\,, \label{psi_2}
\end{align}
\end{subequations}
where $\mathcal{R}(k)$ is the reflection amplitude and $\mathcal{T}_1(k),\mathcal{T}_2(k)$ are the sought transmission amplitudes.
When we substitute expressions~\eqref{psi_-12} into boundary conditions~\eqref{bc switch}, we obtain the set of conditions
\begin{subequations}\label{bc_sub switch}
\begin{gather}
\i k(-1+\mathcal{R}(k)+\beta\mathcal{T}_2)+\alpha\sum_{j=1}^{n}\phi_j'(0)=0\,; \label{bc1 switch} \\
\i k(\mathcal{T}_1(k))+\alpha\beta\sum_{j=1}^{n}\phi_j'(0)=0\,; \label{bc2 switch} \\
\beta(1+\mathcal{R}(k))=\mathcal{T}_2(k)\,; \label{bc3 switch} \\
\alpha(1+\mathcal{R}(k))+\alpha\beta\mathcal{T}(k)=\phi_j(0) \qquad\text{for all $j=1,\ldots,n$}\,. \label{bc4 switch}
\end{gather}
\end{subequations}
Let $\Lambda$ be the Dirichlet-to-Neumann function for the graph $\Gamma$, cf.~\eqref{Lambda}, and let $\sigma_0$ have the meaning introduced in~\eqref{sigma0}. The transmission amplitudes $\mathcal{T}_1(k)$ and $\mathcal{T}_2(k)$ on the graph $\Gamma_{\mathbf{ioo}}$ with boundary conditions~\eqref{bc switch} in the vertex $0$ can be obtained by solving the system of equations~\eqref{bc_sub switch}. The result is summarized in Proposition~\ref{Prop. T switch}.

\begin{proposition}\label{Prop. T switch}
When a particle of energy $E=\frac{\hbar^2k^2}{2m}$ comes in the vertex $0$, the transmission amplitudes to the output lines are given by the following formulae:
\begin{itemize}
\item[(i)] If $E\in(0,+\infty)\backslash\sigma_0$, then
\begin{subequations}\label{T switch}
\begin{align}
&\mathcal{T}_1(k)=\frac{-2\alpha^2\Lambda(E)}{\frac{\alpha^2\Lambda(E)}{\beta}+\left(\frac{1}{\beta}+\beta\right)\left(\i k+\beta^2\alpha^2\Lambda(E)\right)}\,, \\
&\mathcal{T}_2(k)=\frac{2}{\frac{\alpha^2\Lambda(E)}{\beta\left(\i k+\beta^2\alpha^2\Lambda(E)\right)}+\frac{1}{\beta}+\beta}\,.
\end{align}
\end{subequations}
\item[(ii)] If $E\in\sigma_0$, then
$$
\mathcal{T}_1(k)=\frac{-2}{\frac{1}{\beta}+\beta+\beta^3}\,, \qquad \mathcal{T}_2(k)=\frac{2\beta^2}{\frac{1}{\beta}+\beta+\beta^3}\,.
$$
\end{itemize}
\end{proposition}
\begin{corollary}
For every $k>0$,
$$
\lim_{\alpha\to\infty}\P_1(k)=
\left\{\begin{array}{cl}
0 & \text{if $E=\frac{\hbar^2k^2}{2m}\in\sigma(H_{\Gamma})\backslash\sigma_0$}\,, \\
\frac{4}{\left(\frac{1}{\beta}+\beta+\beta^3\right)^2}
 & \text{otherwise}\,.
\end{array}\right.
$$
$$
\lim_{\alpha\to\infty}\P_2(k)=
\left\{\begin{array}{cl}
\frac{4}{\left(\frac{1}{\beta}+\beta\right)^2} & \text{if $E=\frac{\hbar^2k^2}{2m}\in\sigma(H_{\Gamma})\backslash\sigma_0$}\,, \\
\frac{4\beta^4}{\left(\frac{1}{\beta}+\beta+\beta^3\right)^2}
 & \text{otherwise}\,.
\end{array}\right.
$$
\end{corollary}
The corollary above can be proven in a similar way as Corollary~\ref{T coro}. Now let us define
$$
P_\beta:=\frac{4}{\left(\frac{1}{\beta}+\beta\right)^2}\,.
$$
We observe that if
\begin{equation}\label{beta}
\beta\in[1/4,2/3]\,,
\end{equation}
then $\beta^4$ is small with respect to $1$, hence
\begin{align*}
E\in\sigma(H_{\Gamma})\backslash\sigma_0 \ &\Rightarrow\  \lim_{\alpha\to\infty}\P_1(k)=0\ll\lim_{\alpha\to\infty}\P_2(k)\approx P_\beta\,, \\
E\in\rho(H_{\Gamma})\cup\sigma_0 \ &\Rightarrow\  \lim_{\alpha\to\infty}\P_1(k)= P_\beta\gg\lim_{\alpha\to\infty}\P_2(k)\,,
\end{align*}
and at the same time $\P_\beta$ is high enough to be easily observed ($\P_\beta>0.2$).
To sum up, if we choose $\beta$ according to~\eqref{beta} and $\alpha$ such that $\alpha\gg1$, the device depicted in Figure~\ref{Fig: switch} works as a spectral separator. If the particle energy $E$ neither belongs to $\sigma(H_{\Gamma})$ nor is close to a certain $\lambda\in\sigma(H_{\Gamma})$, the particle is transmitted to the output $\mathbf{o}_1$. If $E\in\sigma(H_{\Gamma})\backslash\sigma_0$ or $E$ is close to a certain $E\in\sigma(H_{\Gamma})\backslash\sigma_0$, the particle is redirected to the output $\mathbf{o}_2$.
Figure~\ref{Fig: switch 3} illustrates its function for $\beta=1/3$.
\begin{figure}[h]
\begin{center}
\includegraphics[width=6.5cm]{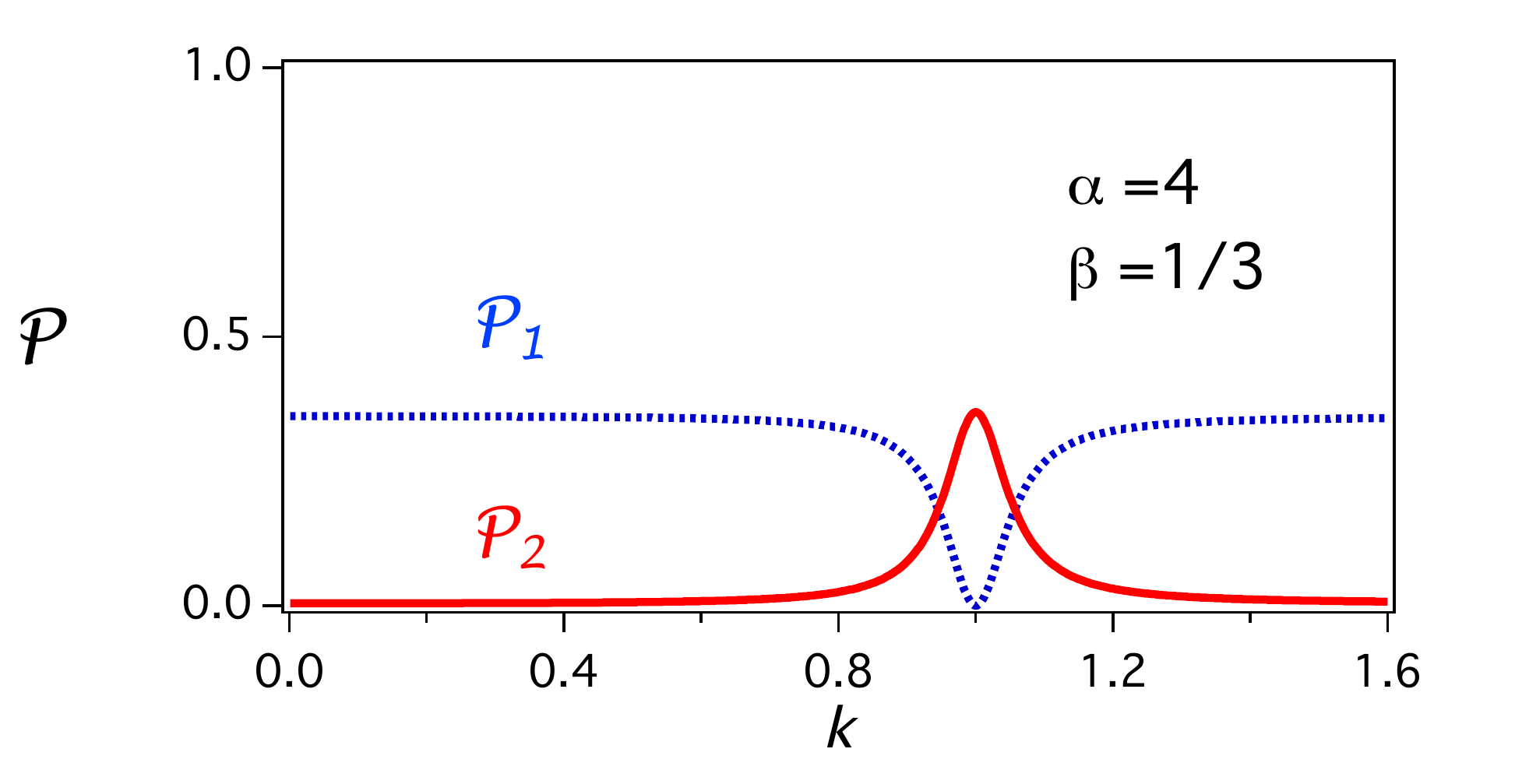}
\caption{An example of transmission characteristics of the separator depicted in Figure~\ref{Fig: switch}. The parameters of the device are chosen as $\beta=1/3$ and $\alpha=4$.}
\label{Fig: switch 3}
\end{center}
\end{figure}

In case that the spectrum of $\Gamma$ is governed by an external field, such as in the case of a loop considered in Section~\ref{Section: Loop}, the device works as a controllable junction which enables to send-out a near-monochromatic pulse of specified spectrum. Finally, if the energy of the incoming particles is fixed, the graph serves as a switch that can turn on and off the flux to a given output line.

\begin{remark}
Following the procedure from paper~\cite{TC12}, one can obtain an approximation of the scale-invariant vertex couplings used in Sections~\ref{Section: Inverse} and \ref{Section: Switch}, similarly as we did in Section~\ref{Section: Approximation}.
\end{remark}

\section{Conclusion}\label{Section: Conclusion}

The key idea 
for achieving strong resonance peaks at desired energies, described in this paper,
consists in attaching a quantum graph $\Gamma$ with convenient spectral properties to the input--output line via a special scale-invariant vertex coupling. The scale-invariant coupling causes the transmission probability along the input--output line, as a function of the particle energy, to resonate at the eigenenergies of the attached graph $\Gamma$. Consequently, the system works as a band-pass filter with 
narrow
passbands located around energies $E\in\sigma(H_{\Gamma})$. It can be also regarded as a spectral analyzer, a device that maps out the spectra of $\Gamma$ through the elastic scattering of a particle off $\Gamma$ with variable incoming energy.

A technically simple concept of controllability is naturally inhered in the model. 
Positions of peaks in transmission characteristics are
controllable by any mode that allows a variation of the spectrum of $\Gamma$. The practically most convenient way is to expose $\Gamma$ to an external field. When the strength of the field is being adjusted, the spectrum of $\Gamma$ is varying, and the 
resonance peaks
are changing their positions accordingly. An example of implementation has been discussed in Section~\ref{Section: Loop}, where $\Gamma$ being a loop in a magnetic field has been considered.

Scale-invariant vertex couplings proved useful already in a previous related work~\cite{TC11,TC12}. In this paper, we applied three different types of these couplings to design three different types of devices: a band-pass filter, a band-stop filter, and a spectral separator. It becomes increasingly evident that scale-invariant vertex couplings can serve as a core component of many simple quantum systems with various scattering properties.

Effects in quantum systems are often experimentally studied using classical waves~\cite{SS90,St99}. For instance,
the behavior of wave functions in quantum graphs is analogical to the behavior of waves in microwave networks~\cite{HBPSZS04}. Therefore, possible applications of our result are not limited to quantum mechanics. The spectral filtering effect could be observed also in various classical systems, such as in optical fibre networks, waveguides and optical laser systems.

\begin{acknowledgments}
The authors thank Pavel Exner for helpful comments.
This research was supported by the Japan Ministry of Education, Culture, Sports, Science and Technology under the Grant number 24540412.
\end{acknowledgments}

\medskip

\end{document}